\documentclass{nm}
\renewcommand\thisnumber{x}
\renewcommand\thisyear {200x}

\renewcommand\thisvolume{xx}

\usepackage{charter}
\usepackage[charter]{mathdesign}

\usepackage{srcltx,graphicx}
\usepackage{amsmath, amssymb, amsthm}
\usepackage{color}
\usepackage{lscape}
\usepackage{float}
\usepackage{bm}
\usepackage{multirow}
\usepackage{psfrag}
\usepackage[hang]{subfigure}
\usepackage{ltxtable}

\newcommand\bbR{\mathbb{R}}
\newcommand\bbN{\mathbb{N}}
\newcommand\bbC{\mathbb{C}}
\newcommand\bbG{\mathbb{G}}
\newcommand\bx{\bm{x}}
\newcommand\bu{\bm{u}}
\newcommand\be{{e}}
\newcommand\bn{\bm{n}}
\newcommand\bw{\bm{w}}

\newcommand\bk{\bm{k}}
\newcommand\bq{\bm{q}}

\newcommand\bzero{{0}}
\newcommand\bA{ {\bf A}}
\newcommand\bB{ {\bf B}}

\newcommand\bD{ {\bf D}}
\newcommand\bE{ {\bf E}}

\newcommand\bI{ {\bf I}}

\newcommand\bK{ {\bf K}}

\newcommand\bM{ {\bf M}}

\newcommand\bP{ {\bf P}}
\newcommand\bPb{{{\bf P}_b}}
\newcommand\bPp{{{\bf P}_p}}
\newcommand\bQ{ {\bf Q}}

\newcommand\bS{ {\boldsymbol{S}}}
\newcommand\bT{ {\bf T}}
\newcommand\bU{ {\boldsymbol{U}}}
\newcommand\bV{ {\boldsymbol{V}}}

\newcommand\bv{\bm{v}}
\newcommand\dd{\,\mathrm{d}}
\newcommand\bxi{\bm\xi}
\newcommand\balpha{\alpha}
\newcommand\bbeta{\beta}
\newcommand\bLambda{{\bf \Lambda}}

\newcommand\imag{\boldsymbol{\mathrm{i}}}

\newcommand\mE{{\mathcal{E}}}

\newcommand\Healphau{\text{He}_{\balpha}^{[\bu,\theta]}}

\newcommand\Halphau{\mathcal{H}_{\balpha}^{[\bu,\theta]}}
\newcommand\Hbetau{\mathcal{H}_{\bbeta}^{[\bu,\theta]}}
\newcommand\Hgammau{\mathcal{H}_{\gamma}^{[\bu,\theta]}}

\newcommand\omegau{\omega^{[\bm{u},\theta]}}

\newcommand\pd[2]{\dfrac{\partial {#1}}{\partial {#2}}}

\def\pkgnmused{1}

\begin{document}


\markboth{Y.-N. Di, Y.-W. Fan, R. Li and L.-C. Zheng}{Linear Stability}
\title{Linear Stability of Hyperbolic Moment Models for Boltzmann
  Equation}


%
%
%
\author[Y.-N. Di, Y.-W. Fan, R. Li and L.-C. Zheng]{Yana
  Di\affil{1}, Yuwei Fan\affil{2}, Ruo
  Li\affil{3}\comma\corrauth and Lingchao Zheng\affil{2}}
\address{\affilnum{1} LSEC, Institute of Computational Mathematics and
  Scientific/Engineering Computing, NCMIS, AMSS, Chinese Academy
  of Sciences, Beijing, China.\\
  \affilnum{2} School of Mathematical Sciences, Peking
  University, Beijing, China.\\
  \affilnum{3} HEDPS \& CAPT, LMAM \& School of Mathematical Sciences,
  Peking University, Beijing, China.} 
\emails{ {\tt yndi@lsec.cc.ac.cn} (Yana Di), {\tt ywfan@pku.edu.cn}
  (Yuwei Fan), {\tt rli@math.pku.edu.cn} (Ruo Li), {\tt
    lczheng@pku.edu.cn} (Lingchao Zheng) }

\begin{abstract}
  Grad's moment models for Boltzmann equation were recently regularized
  to globally hyperbolic systems, and thus the regularized models attain
  local well-posedness for Cauchy data. The hyperbolic regularization
  is only related to the convection term in Boltzmann equation. We in
  this paper studied the regularized models with the presentation of
  collision terms. It is proved that the regularized models are linearly
  stable at the local equilibrium and satisfy Yong's first
  stability condition with commonly used approximate collision terms,
  and particularly with Boltzmann's binary collision model. 
\end{abstract}

\if\pkgnmused1



\maketitle

\fi

\section{Introduction}
Boltzmann equation \cite{Boltzmann} is the most important kinetic
equation, governing the movement of a particle system, particularly
the gas particles. Since the distribution function in the Boltzmann
equation is in very high dimension, Grad \cite{Grad} purposed the
famous moment method for gas kinetic theory to reduce the kinetic
equation into low-dimensional models. In more than half a century,
Grad's moment equations were suffered by the lack of hyperbolicity
\cite{Muller,Grad13toR13}. Only very recently, in \cite{Fan, Fan_new}, 
the authors revealed the underlying reason that Grad's moment equations
lost its hyperbolicity during the model reduction, and 
purposed new reduced models of Boltzmann equation. The new models are
referred to as globally Hyperbolic Moment Equations (HME) hereafter, which
are symmetric quasi-linear systems \cite{framework} with global
hyperbolicity.

As new models for fluid dynamics, one may prefer to carry out studies
on some fundamental mathematical properties on HME before further
numerical applications. Among these fundamental mathematical
properties, linear stability is one of the most important points
\cite{Bobylev,RosenauCE,Struchtrup2003} for a system to be applied in
numerical experiments. It should be noted that the linear stability is not
automatically attained for models in fluid dynamics. For instance,
famous Burnett equations and super-Burnett equations are discovered
not linearly stable \cite{Bobylev,Struchtrup}, and thus are ill-posed
and rarely have practical applications.

Except for linear stability, Yong proposed the called Yong's first 
stability condition \cite{yongdissertation,yong1999singular}, 
for nonlinear first-order hyperbolic systems with
source term. With this stability condition, a formal asymptotic
approximation of the initial-layer solution to the nonlinear problem
has been constructed \cite{yong1999singular}. Furthermore, with
some regularity assumption of the solution, the existence of classical
solutions is guaranteed in the uniform time interval. The stability
condition is essential for the nonlinear first-order hyperbolic
system. And in \cite{yongdissertation,yong1999singular}, several
classical models have been verified to satisfy the stability condition.

In this paper, we focus on the linear stability analysis of HME at
local equilibrium and Yong's first stability condition. The collision
term under consideration includes the commonly used approximate
formations, such as BGK model \cite{BGK}, ES-BGK model \cite{Holway},
Shakhov model \cite{Shakhov} and the original Boltzmann's collision
term \cite{Boltzmann}, particularly the binary collision term
\cite{Cowling, Grad1949note}. We prove that both HME and Ordered
globally Hyperbolic Moment Equations(OHME) are linearly stable at local
equilibrium for all the four collision models, and satisfy Yong's
first stability condition.

We start with a brief review of HME and the collision term to be
considered. The globally hyperbolic regularization enables us to write
HME into an elegant quasi-linear form. It is essential to expand the
distribution function at the local equilibrium, where the collision
term vanishes. This property provides us some additional equalities which
significantly simplify the linear stability analysis. For the binary
collision model, the symmetry of the collision plays an important
role, which indicates some induced symmetry in the Jacobian of the
collision term. With some linear algebra, we proved that HME is linear
stable at local equilibrium for all the four collision models. This
proof is not trivial noticing that HME we are studying is for
arbitrary order.

For Yong's first stability condition, the third inequality plays a
major role. We verified this inequality by applying the results in the
linear stability analysis, together with some linear algebraic
technique. In such sense, Yong's first stability condition can be
regarded as an enhanced version of linear stability for nonlinear
balance laws.

OHME, first proposed in \cite{Fan2015}, is the hyperbolic version of
ordered Grad's moment system, which includes the well-known Grad's 13
moment system. Since OHME can be derived from HME, the linear
stability of OHME at the local equilibrium is deduced from that of
HME, as well as Yong's first stability condition.

The rest of the paper is organized as following.  Section
\ref{sec:Preliminaries} presents a brief introduction of the linear
stability and some useful linear algebraic results.  The Boltzmann
equation and Grad's moment method, together with the globally
hyperbolic moment system are reviewed in Section \ref{sec:hme}.  In
Section \ref{sec:stability}, four Boltzmann collision terms are
studied, and the linear stability of HME at local equilibrium is
rigorous proved. In Section \ref{sec:Yong}, Yong's first stability
condition is verified for HME.  We extend the results in Section
\ref{sec:stability} and Section \ref{sec:Yong} to OHME and prove that OHME
is also linearly stable at local equilibrium and satisfies Yong's
first stability condition in Section \ref{sec:grad13}.  The paper ends
with a conclusion.


\section{Preliminaries} \label{sec:Preliminaries}
\subsection{Linear stability}
Let us consider the linear PDEs with source term as
\begin{equation}\label{eq:balancelaw}
    \pd{\bU}{t}+\sum_{d=1}^D\bA_d\pd{\bU}{x_d}=\bQ\bU,
\end{equation}
where the matrices $\bA_d$, $d=1,\dots,D$, and $\bQ$ are constant.
Following \cite{Struchtrup}, we assume the solution is plane waves of the form 
\begin{equation}
    \bU = {\bU}_*\exp\left( \imag(\Omega t-\bk^T\bx) \right),
\end{equation}
where $\imag$ is the imaginary unit, ${\bU}_*$ is the complex
amplitude of the wave, $\Omega$ is its frequency and $\bk$ is its wave
number. Here we use complex variables for convenience, and only the
real parts of the expressions for the $\bU$ are relevant. The
equation \eqref{eq:balancelaw} can be rewritten as
\begin{equation}
    \left(\imag\Omega\bI-\sum_{d=1}^D\imag k_d\bA_d-\bQ\right){\bU}_*=0,
\end{equation}
where $\bI$ is the identity matrix.
The existence of a nontrivial solution $\bU_*$ of the equation
requires the coefficient matrix to be singular
\begin{equation} \label{singularity}
	\det\left(\Omega \bI - \sum_{d=1}^D k_d\bA_d+\imag\bQ\right)=0.
\end{equation}
This gives us the dispersion relation between $\Omega$ and $\bk$.

Considering a disturbance in space, the wave number $\bk$ is real
and the frequency is complex $\Omega=\Omega_r(\bk)+\imag\Omega_i(\bk)$.
Then the plane wave solutions have the form 
\[
    \bU=\bU_*\exp(-\Omega_i(\bk)t)\exp(\imag(\Omega_r(\bk)t-\bk^T\bx)).
\]
Note that $\bU_*\exp(-\Omega_i(\bk)t)$ is the local amplitude of $\bU$
as a function of time, and stability requires the local amplitude to
be non-increasing, thus $\Omega_i(\bk) \geq 0$.

If we consider a disturbance in time at a given location, the
frequency $\Omega$ is real and the wave number is complex
$k=k_r(\Omega)+\imag k_i(\Omega)$, where we consider this problem for
one-dimensional processes following \cite{Struchtrup,Struchtrup2003}.
Then the plane wave solutions is
\[
    \bU=\bU_*\exp(k_{i}(\Omega)x)\exp(\imag(\Omega t-k_{r}(\Omega)x)).
\]
Here $\bU_*\exp(k_{i}(\Omega)x)$ is the amplitude of $\bU$ at the
point $x$. To be a stable solution, which is a wave traveling in
positive $x$ direction($k_{r}>0$), it requires a non-increasing
amplitude ($k_{i} \leq 0$), and vice versa, thus $k_{r} k_{i} \leq 0$.

\begin{definition}[Stability]
  The system \eqref{eq:balancelaw} is stable in time if
  $\Omega_i(\bk)\geq0$ for each $\bk\in\bbR^D$; it is stable in space
  for one-dimensional processes if $k_{r}(\Omega)k_{i}(\Omega)\leq0$
  for each $\Omega\in\bbR^+$.
\end{definition}

\subsection{Yong's first stability condition}
In \cite{yong1999singular}, Yong developed a singular perturbation
theory for initial-value problems of nonlinear first-order hyperbolic
system with stiff source term in several space variables, and proposed
the stability condition. Under the stability condition, a formal
asymptotic approximation of the initial-layer solution to the
nonlinear problem are constructed. Moreover, with some regularity
assumption on the solution, the existence of classical solutions is
guaranteed in uniform time interval. The stability condition is
fundamental for the nonlinear first-order hyperbolic system with the
form
\begin{equation}
    \pd{\bU}{t}+\sum_{d=1}^D\bA_d(\bU)\pd{\bU}{x_d}=\bS(\bU),\quad 
    \bU\in\bbG\subset\bbR^n.
\end{equation}
Let $\bQ=\pd{\bS}{\bU}$, and define the equilibrium manifold
\[
    \mE:=\{\bU\in\bbG\mid\bS(\bU)=0\}.
\]
The stability condition in \cite{yong1999singular} reads {\it
\begin{enumerate}
    \item There is an invertible $n \times n$ matrix $\bP(\bU)$ and an
        invertible $r\times r$ matrix $\hat{\bQ}(\bU)$, defined on the
        equilibrium manifold $\mE$, such that
        \begin{equation}\label{eq:Yong1}
            \bP(\bU)\bQ(\bU) =
            \left(\begin{array}{cc}
                0 & 0\\
                0 & \hat{\bQ}(\bU)
            \end{array}\right) 
            \bP(\bU)\quad \text{ for }
            \bU\in\mE.
        \end{equation}
  \item There is a symmetric positive definite matrix $\bA_0(\bU)$
      such that
      \begin{equation}\label{eq:Yong2}
          \bA_0(\bU)\bA_d(\bU)=\bA_d^T(\bU)\bA_0(\bU),
          \qquad\bU\in\bbG,~ d=1,\dots,D.
      \end{equation}
  \item The hyperbolic part and the source term are coupled in the sense
      \begin{equation}\label{eq:Yong3}
          \bA_0(\bU)\bQ(\bU)+\bQ(\bU)^T\bA_0(\bU)\leq-\bP(\bU)^T
          \left(\begin{array}{cc}
              0 & 0\\
              0 & \mathbf{I}_r
          \end{array}\right)
          \bP(\bU).
      \end{equation}
\end{enumerate}}
The first condition requires that the source term is dissipation or
relaxation, and the second condition guarantees that the hyperbolic
part is a symmetric hyperbolic system. The third condition specifies
how the hyperbolic part and the source term can be coupled, which is
the key condition to the stability of the solution.

\subsection{Two lemmas}
At the end of the section, we give two useful lemmas in linear
algebra for usage later on.
\begin{lemma} \label{eigenvaluenonnegative}
  Matrices $\bA$, $\bB \in \bbR^{n\times n}$ are symmetric, and $\bB$
  is negative semi-definite, then each eigenvalue of the matrix $\bA
  -\imag\bB$ has a non-negative imaginary part.
\end{lemma}
\begin{proof}
  We prove it by contradiction. Suppose that $\lambda = a + b \imag$, $a$,
  $b \in \bbR$, is an eigenvalue of matrix $\bA -\imag\bB$, and $b<0$,
  with the corresponding eigenvector $\bv \in \bbC^n$, then
  \[    
      [(\bA - a \bI) - \imag (\bB + b \bI)] \bv = \bm{0}.
  \]
  Denote $\bv$ by $\bv=\bv_r + \imag \bv_i$, $\bv_r$, $\bv_i \in \bbR^n$.
  Multiplying the upper equation by $\overline{\bv} = \bv_r - \imag
  \bv_i$, we obtain
  \[
      \overline{\bv}^T (\bA - a \bI) \bv - \imag \overline{\bv}^T (\bB + b
      \bI) \bv = 0.
  \]
  Noticing that $\overline{\bv}^T (\bA - a \bI) \bv$,
  $\overline{\bv}^T (\bB + b \bI) \bv \in \bbR$, we have
  $\overline{\bv}^T (\bB + b \bI) \bv = 0$. Direct calculations yield
  that $\bv_r^T (\bB + b \bI) \bv_r + \bv_i^T (\bB + b \bI) \bv_i =
  0$. Since $\bB + b \bI$ is symmetric negative definite, $\bv_r$ and
  $\bv_i$ have to vanish, and thus $\bv=\bm{0}$. This contradiction
  ends the proof.
\end{proof}

\begin{lemma} \label{eigenvaluek} 
    Matrices $\bA$, $\bB \in \bbR^{n\times n}$ are symmetric, and
    $\bB$ is negative semi-definite. Let $k = k_r + i k_i \in \bbC$,
    $k_r$, $k_i \in \mathbb{R}$ be the solution of 
    $\det(k \bA - \imag \bB -\lambda \bI) = 0$, for any given
    $0<\lambda \in \mathbb{R}$, then $k_r k_i\leq 0$.
\end{lemma}
\begin{proof}
    Let $\bv \in \bbC^n$, $\| \bv \| \neq 0$ be an vector s.t.
    $(k \bA - \imag \bB - \lambda \bI) \bv = \bm{0}$, then
    $\overline{\bv}^T (k_r \bA - \lambda \bI) \bv + \imag \overline{\bv}^T
    (k_i \bA - \bB) \bv = 0$, which indicates
    $k_r \overline{\bv}^T \bA \bv = \lambda \overline{\bv}^T \bv > 0$
    and $k_i \overline{\bv}^T \bA \bv = \overline{\bv}^T \bB \bv \leq
    0$. Thus, $k_rk_i\leq 0$.
\end{proof}


\section{HME for Boltzmann Equation}
\label{sec:hme}
Let us denote the distribution function in gas kinetic theory by
$f(t,\bx,\bxi)$ describing the probability density to find a particle
at space point $\bx$ and the time $t$ with velocity $\bxi$ in
$D$-dimensional space. The macroscopic density $\rho$, flow velocity
$\bu$, temperature $T$, pressure tensor $P=(p_{ij})_{D\times D}$,
stress tensor $\Sigma=(\sigma_{ij})_{D\times D}$ and heat flux $\bq$
are related to the distribution function by
\begin{equation}
    \begin{aligned}
        \rho(t,\bx)&=\int_{\bbR^D}f(t,\bx,\bxi)\dd\bxi,   &
        P&=\int_{\bbR^D}(\bxi-\bu)\otimes(\bxi-\bu)f(t,\bx,\bxi)\dd\bxi,\\
        \rho(t,\bx)\bu(t,\bx)&=\int_{\bbR^D} \bxi f(t,\bx,\bxi)\dd\bxi,   &
        \bq&=\frac{1}{2}\int_{\bbR^D}|\bxi-\bu|^2(\bxi-\bu)f(t,\bx,\bxi)\dd\bxi,\\
        D \rho RT
        &=\int_{\bbR^D}|\bxi-\bu|^2 f(t,\bx,\bxi)\dd\bxi, &
        \Sigma&=P-p\bI, 
    \end{aligned}
\end{equation}
where $p= \displaystyle \frac{1}{D}\sum_{d=1}p_{dd}=\rho RT$ is
pressure, and the constant $R$ stands for the gas constant. For
convenience, use $\theta(t,\bx) = R T(t,\bx)$ to simplify the
notations.

The distribution function $f(t,\bx,\bxi)$ is governed by the Boltzmann
equation\cite{Boltzmann}
\begin{equation} \label{BoltzmannEQ}
  \pd{f}{t}+\bxi\cdot \nabla_{\bx} f = Q(f,f),
\end{equation}
where the right hand side $Q(f,f)$ is the collision term, which models
the interaction among particles at the position $\bx$ and time $t$.  
The collision term is assumed to have only $1$, $\bxi$ and
$|\bxi|^2$ as locally conserved quantities, saying
\begin{equation}\label{eq:conservedquantities}
    \int_{\bbR^D}Q(f,f)(1,\bxi,|\bxi|^2)^T\dd\bxi=\boldsymbol{0},
\end{equation}
and 
\begin{equation}\label{eq:conservedquantities2}
    \text{if } \int_{\bbR^D}Q(f,f)\psi(\bxi)\dd\bxi=0, 
    \text{ for all } f, \text{ then }
    \psi(\bxi)=a+\boldsymbol{b}^T\bxi+c|\bxi|^2.
\end{equation}
The collision term is also assumed that
\begin{equation}\label{eq:equilibriummanifold}
    Q(f,f)=0\Rightarrow f=f_{eq},
\end{equation}
where $f_{eq}$ is the local equilibrium 
\begin{equation} \label{Maxwellian}
    f_{eq}(t,\bx,\bxi)=\frac{\rho(t,\bx)}{[2\pi\theta(t,\bx)]^{D/2}}
    \exp \left( -\frac{|\bxi-\bu(t,\bx)|^2} {2\theta(t,\bx)} \right).
\end{equation}

The binary collision term\cite{Cowling,Grad} is commonly used to model
the dilute gas, and has a quadratic form
\begin{equation} \label{binarycollsion}
    Q(f,f)=\int_{\bbR^D}\int_{S_+^{D-1}}
    (f'f'_1-ff_1)B(|\bxi-\bxi_1|,\sigma)\dd\bn\dd\bxi_1,
\end{equation}
where $S_+^{D-1}$ is the upper half sphere, $B(|\bxi-\bxi_1|,\sigma)$
is the collision kernel, and $\sigma$ is a function of $\bn$, $\bxi$
and $\bxi_1$, depending on the type of particles. 
In \eqref{binarycollsion},
\[
f=f(t,\bx,\bxi),~~ f_1=f(t,\bx,\bxi_1),~~ f'=f(t,\bx,\bxi'),~~
f'_1=f(t,\bx,\bxi'_1),
\]
where $\bxi$ and $\bxi_1$ are the velocities of two particles before
collision, $\bxi'$ and $\bxi'_1$ are their velocities after collision,
and $\bn$ is the direction between their centers of mass.
The specific expressions of $B(|\bxi-\bxi_1|,\sigma)$ and $\sigma$ are
not concerned in this paper.

As simplifications of the binary collision, researchers proposed some
alternative collision models to approximate the binary collision
model, such as BGK model\cite{BGK}, Shakhov model\cite{Shakhov} and
ES-BGK model\cite{Holway}. We list these models below for later usage:
\begin{itemize}
    \item Bhatnagar-Gross-Krook(BGK) model \cite{BGK}:
        \begin{equation} \label{BGKeq}
            Q(f,f)=\frac{1}{\tau}(f_{eq}-f),
        \end{equation}
        where $\tau$ is relaxation time.
    \item Shakhov model \cite{Shakhov}:
        \begin{equation} \label{ShakhovEQ}
            Q(f,f)=\frac 1 \tau (f_S-f),
        \end{equation}
        where
        \begin{equation*}
            f_S(t,\bx,\bxi)=f_{eq}(t,\bx,\bxi)\left( 1 + \dfrac{(1-\Pr)
            \bq^T (\bxi - \bu)}
            {(D+2)\rho\theta^2} \left(\dfrac{|\bxi-\bu|^2}{\theta}-(D+2)\right)
            \right),
        \end{equation*}
        where $\Pr$ is the Prandtl number, which is $2/3$ for
        monatomic gas.
\item ES-BGK model \cite{Holway}:
    \begin{equation} \label{ES-BGKeq}
        Q(f,f)=\frac{\Pr}{\tau}(f_G-f),
    \end{equation}
    where
    \begin{equation*}
        f_G=\dfrac{\rho}{\sqrt{
            \mathrm{det}(2\pi\bLambda)}}\exp\left(-\dfrac{1}{2}
        (\bxi-\bu)^T \bLambda^{-1}(\bxi-\bu) \right),
    \end{equation*}
    where $\bLambda = (\lambda_{ij}) \in \bbR^{D\times D}$ is a
    symmetric positive definite matrix with entries as $\lambda_{ij} =
    \dfrac{\theta \delta_{ij}}{\text{Pr}} + \left( 1 - \dfrac 1
    {\text{Pr}} \right) \dfrac{p_{ij}}{\rho}$, $i$, $j=1, \cdots, D$,
    and $\delta_{ij}$ is Kronecker delta symbol.
\end{itemize}
All of the four collision models satisfy the relationship
\eqref{eq:conservedquantities}, \eqref{eq:conservedquantities2} and
\eqref{eq:equilibriummanifold}.

In 1949, Grad proposed the well-known Grad's moment method\cite{Grad}
to derive moment equations from the Boltzmann equation. The key point
is to expand the distribution function around the local Maxwellian
into Hermite series as
\begin{equation}
	\label{fexpand}
	f(t,\bx,\bxi)=\sum_{\balpha\in\bbN^D} 
	f_{\balpha}(t,\bx) \Halphau(\bxi),
\end{equation}
where $\balpha$ is a $D$-dimensional multi-index, $\Halphau(\bxi)$ is
the basis function, defined by
\begin{equation} \label{Halphau}
    \Halphau(\bxi)=\Healphau(\bxi)\omegau(\bxi),
    \quad 
    \omegau(\bxi)=\frac{f_{eq}}{\rho}
    =\frac{1}{[2\pi\theta]^{D/2}}
    \exp\left( -\frac{|\bxi-\bu|^2}{2\theta} \right),
\end{equation}
where 
\begin{equation}
    \Healphau(\bxi)=
    \frac{1}{\omegau(\bxi)}\prod_{d=1}^D \pd{^{\alpha_d}}{\xi_d^{\alpha_d}}\omegau(\bxi)
    ,\quad \balpha\in\bbN^D.
\end{equation}
Due to the orthogonality of the basis function, we have
\cite{Qiao,FanDissertation}
\begin{equation}\label{eq:def_falpha}
    f_{\balpha}=\frac{\theta^{|\balpha|}}{\balpha!}\int_{\bbR^D}f\Healphau(\bxi)\dd\bxi,
\end{equation}
where $|\balpha|=\sum_{d=1}^D\alpha_d$, and $\balpha!=\prod_{d=1}^D\alpha_d!$. 
Particularly, we have for $i,j=1,\cdots,D$
\begin{equation}
    \begin{aligned}
        &f_0=\rho,\quad f_{e_i}=0, \quad \sum_{d=1}^Df_{2e_d}=0,\\
        &p_{ij}=p\delta_{ij}+(1+\delta_{ij})f_{e_i+e_j},\quad
        q_i=2f_{3e_i}+\sum_{d=1}^Df_{e_i+2e_d},
    \end{aligned}
\end{equation}
where $e_i$, $i=1,\cdots,D$ is unit multi-index with its $i$-th entry
to be $1$.

Substituting the expansion \eqref{fexpand} into the Boltzmann equation
\eqref{BoltzmannEQ}, and matching the coefficients of the basis
function $\Halphau(\bxi)$, we can obtain the governing equation of
$\bu$, $\theta$ and $f_{\alpha}$, $\alpha\in\bbN^3$. However, the
resulting system contains infinite number of equations. A cut-off and
moment closure are required. Choosing a positive integer
$3\leq M\in\bbN$, and discarding all the equations including
$\pd{f_{\alpha}}{t}$, $|\alpha|>M$, and setting $f_{\alpha}=0$,
$|\alpha|>M$ to closure the residual system, we can obtain $M$-th
order Grad's moment system as
\begin{equation}
	\begin{aligned}
		\pd{f_{\balpha}}{t}&+\sum_{d=1}^D\left( \theta \pd{f_{\balpha-\be_d}}{x_d}+u_d\pd{f_{\balpha}}{x_d}+
		(1-\delta_{|\balpha|,M})
		(\balpha_d+1)\pd{f_{\balpha+\be_d}}{x_d} \right) \\
		+\sum_{k=1}^Df_{\balpha-\be_k}\pd{u_k}t&+\sum_{k,d=1}^D\pd{u_k}{x_d}(\theta f_{\balpha-\be_k-\be_d} +u_d f_{\balpha-\be_k}+
		(\balpha_d+1)f_{\balpha-\be_k+\be_d }) \\
		+\frac 12 \sum_{k=1}^D f_{\balpha-2\be_k}\pd{\theta}t&+\sum_{k,d=1}^D\frac 12 \pd{\theta}{x_d}(\theta f_{\balpha-2\be_k-\be_d}+u_df_{\balpha-2\be_k}+
		(\balpha_d+1)f_{\balpha-2\be_k+\be_d}) \\
		&=S_{\balpha},\quad |\balpha|\leq M,
	\end{aligned}
	\label{Gradeq}
\end{equation}
where 
\begin{equation}\label{eq:def_Salpha}
    S_{\balpha}=
    \frac{\theta^{|\balpha|}}{\balpha!}\int_{\bbR^D}Q(f,f)\Healphau(\bxi)\dd\bxi.
\end{equation}
It is well-known that Grad's moment system lacks of global
hyperbolicity \cite{Muller} and it was found recently that it is not
hyperbolic even around the local Maxwellian \cite{Grad13toR13}. The
globally hyperbolic regularization proposed in \cite{Fan,Fan_new}
essentially fixes this drawback and yields the globally hyperbolic
moment equations (HME) as
\[
	\begin{aligned}
		\pd{f_{\balpha}}{t}&+\sum_{d=1}^D\left( 
		\theta \pd{f_{\balpha-\be_d}}{x_d}+u_d\pd{f_{\balpha}}{x_d}+
		(1-\delta_{|\balpha|,M})
		(\balpha_d+1)\pd{f_{\balpha+\be_d}}{x_d} \right) \\
		+\sum_{k=1}^Df_{\balpha-\be_k}\pd{u_k}t&+
		\sum_{k,d=1}^D\pd{u_k}{x_d}(\theta f_{\balpha-\be_k-\be_d} +u_d f_{\balpha-\be_k}+
		(1-\delta_{|\balpha|,M})
		(\balpha_d+1)f_{\balpha-\be_k+\be_d }) \\
		+\frac 12 \sum_{k=1}^D f_{\balpha-2\be_k}\pd{\theta}t&
		+\sum_{k,d=1}^D\frac 12 \pd{\theta}{x_d}(\theta f_{\balpha-2\be_k-\be_d}+u_df_{\balpha-2\be_k}+
		(1-\delta_{|\balpha|,M})
		(\balpha_d+1)f_{\balpha-2\be_k+\be_d}) \\
		&=S_{\balpha},\quad |\balpha|\leq M,
	\end{aligned}
	\label{HME}
\]
where $(\cdot)_{\alpha}$ is taken as zero if any component of $\alpha$
is negative. To simplify the notations, we introduce the ordering
relation on $\bbN^D$.
\begin{definition}[Graded reverse lexicographic]
    An ordering relaxation on $\bbN^D$ is called graded reverse
    lexicographic ordering $\prec$ if for any $\balpha,
    \bbeta\in\bbN^D$
    \[
        \begin{aligned}
            \balpha\prec\bbeta \Longleftrightarrow&
            |\balpha|\le|\bbeta| ~\text{or}~\\
            &|\balpha|=|\bbeta|, ~
            \text{and}~\exists i (1\leq i\leq D), ~\text{s.t.}
            ~\alpha_i>\beta_i,~\alpha_j=\beta_j (i < j \leq D).
        \end{aligned}
    \]
\end{definition}
With this ordering, we adopt the multi-indices as the subscripts of
vectors and matrices since now on, sorting the multi-indeices by the
graded reverse lexicographic ordering $\prec$. Let $N$ to be all the
multi-indices not greater than $Me_D$, which is the total number of
the equations in $M$-th order Grad's moment system. For a vector
$\bw\in\bbR^N$, $w_\alpha$ stands for the entry with $\alpha$ as
subscript, and for a matrix $\bD \in \bbR^{N \times N}$,
$D_{\alpha,\beta}$ stands for the entry with row index $\alpha$ and
column index $\beta$.

Following the notations in \cite{Fan_new}, define $\bw\in\bbR^N$ and
\begin{align} \label{bwdefine}
	w_{\balpha}= \left\{ 
	\begin{aligned}
		&\rho,    &&\balpha=\bzero, \\
		&u_i,   && \balpha=\be_i,~i=1,\cdots,D, \\
		&\frac{p_{ij}}{1+\delta_{ij}},  &&
                \balpha=\be_i+\be_j,~i,j=1,\cdots,D, \\
		&f_{\balpha},   &&3\leq|\balpha|\leq M.
	\end{aligned}
	\right.
\end{align}
The HME \eqref{HME} can be written into quasi-linear form\cite{framework}:
\begin{equation} \label{HMEquasilinear}
	\bD \pd{\bw}{t}+\sum_{d=1}^D \bM_d\bD\pd{\bw}{x_d}=\bS,
\end{equation}
where the coefficient matrices $\bD$, $\bM_d$ are defined as\cite{FanDissertation}
\begin{align}
	\begin{split} \label{bD}
		\bD&=\bI+\sum_{|\balpha|\leq M}\Big(
		\sum_{d=1}^Df_{\balpha-\be_d}\bE_{\balpha,\be_d}
		-\frac{\theta}{2\rho}\sum_{d=1}^Df_{\balpha-2\be_d}\bE_{\balpha,\bzero}\\
		&\qquad\qquad+H(|\balpha|-3)\frac
        1{D\rho}\left(\sum_{d=1}^Df_{\balpha-2\be_d}\right)\sum_{k=1}^D\bE_{\balpha,2\be_k}\Big)
		-\sum_{d=1}^D \bE_{\be_d,\be_d},
	\end{split}\\
	\bM_d &= \sum_{|\balpha|\leq M}(\theta \bE_{\balpha,\balpha-\be_d}
	+u_d\bE_{\balpha,\balpha}+(1-\delta_{|\balpha|,M})(\alpha_d+1)\bE_{\balpha,\balpha+\be_d}),
	\label{bM}
\end{align}
where $\bI$ is the identity matrix and $\bE_{\balpha,\bbeta}$ is zero
matrix if any component of $\balpha,\bbeta$ is negative or
$|\bbeta|>M$, and is the matrix with all its entries to be $0$, except
for the only entry with row index $\balpha$ and column index $\beta$
to be $1$. The Heaviside step function $H(x)$ is defined as
\begin{equation*}
    H(x)= \left \{
        \begin{aligned}
            0, \quad x<0, \\
            1, \quad x\geq 0.
        \end{aligned} \right.
\end{equation*}
As pointed out in \cite{FanDissertation}, $\bD$ is a lower triangular
matrix with all diagonal entries nonzero thus invertible, and its
inverse is
\begin{align} \label{bD-1}
	\begin{aligned}
		\bD^{-1}=\bI-\sum_{|\balpha|\leq M}&\Big(\sum_{d=1}^D\frac{f_{\balpha-\be_d}}{\rho}\bE_{\balpha,\be_d}
		+H(|\balpha|-3)\frac 1{D\rho}\left(\sum_{d=1}^Df_{\balpha-2\be_d}\right)\sum_{k=1}^D\bE_{\balpha,2\be_k}\Big)\\
		&+\sum_{d=1}^D \frac 1\rho \bE_{\be_d,\be_d} +\frac{\theta}2\sum_{d=1}^D \bE_{2\be_d,\bzero}.
	\end{aligned}
\end{align}

Noticing \eqref{eq:conservedquantities} and \eqref{bD-1}, we obtain
$\left( \bD^{-1}-\bI \right)\bS=\boldsymbol{0}$, thus
\begin{equation} \label{bmS}
\bD^{-1}\bS=\bS.
\end{equation}
Hence, the HME \eqref{HMEquasilinear} can be reformulated as
\begin{equation}
	\label{HMEmatrix}
	\pd{\bw}{t}+\sum_{d=1}^D \bA_d\pd{\bw}{x_d}=\bS,
\end{equation}
where $\bA_d=\bD^{-1}\bM_d\bD$. 


\section{Linear Stability of HME}
\label{sec:stability}
Now we begin to investigate the linear stability of the HME at the
thermodynamic equilibrium. First we linearize the HME into linear
balance laws at a local Maxwellian given by $\rho_0$,
$\bu_0=\boldsymbol{0}$, and $\theta_0$. Let us introduce the
dimensionless variables $\bar{\rho}$, $\bar{\theta}$, $\bar{\bu}$,
$\bar{p}$, $\bar{p}_{ij}$ and $\bar{f}_{\alpha}$ as
\begin{equation}\label{eq:dimensionless}
  \begin{aligned}
    &\rho = \rho_0 (1 + \bar{\rho}),\quad u_i = \sqrt{\theta_0} \bar{u}_i,
    \quad \theta = \theta_0 (1 + \bar{\theta}),
    \quad p = p_0 (1 + \bar{p}),\\
    &p_{ij}=p_0(\delta_{ij}+\bar{p}_{ij}),
    \quad f_{\alpha}=\rho_0\theta_0^{\frac{|\alpha|}{2}} \cdot \bar{f}_{\alpha},
    \quad \bx = L\cdot
    \bar{\bx},\quad t = \frac{L}{\sqrt{\theta_0}}\bar{t},
  \end{aligned}
\end{equation}
where $L$ is a characteristic length, $\bar{\bx}$ and $\bar{t}$ are
the dimensionless coordinates and time, respectively.
Let 
\begin{equation}\label{eq:dimensionOfbw}
    \bLambda_0=\sum_{|\alpha|\leq
    M,|\alpha|\neq 1}\rho_0\theta_0^{|\alpha|/2}\bE_{\alpha,\alpha}
    +\sqrt{\theta_0}\sum_{d=1}^D\bE_{e_d,e_d},
\end{equation}
\begin{equation}\label{eq:bw0}
    \bw_0=\left\{ \begin{array}{l}
		1,    \\
		0,    \\
        \frac{\delta_{ij}}{2},\\
		0, 
    \end{array} \right.
    \qquad
    \bar{\bw}=\left\{ \begin{array}{ll}
        \bar{\rho},    &\balpha=\bzero, \\
        \bar{u}_i,   & \balpha=\be_i,~i=1,\cdots,D, \\
        \frac{\bar{p}_{ij}}{1+\delta_{ij}},  & \balpha=\be_i+\be_j,~i,j=1,\cdots,D, \\
        \bar{f}_{\alpha},   &3\leq|\balpha|\leq M,
    \end{array} \right.
\end{equation}
then $\bw=\bLambda_0(\bw_0+\bar{\bw})$.
All the dimensionless variables $\bar{\rho}$, $\bar{\theta}$,
$\bar{\bu}$, $\bar{p}$, $\bar{p}_{ij}$ and $\bar{f}_{\alpha}$ are
small quantities.
Substituting \eqref{eq:dimensionless}, \eqref{eq:dimensionOfbw} and
\eqref{eq:bw0} into the globally hyperbolic moment system
\eqref{HMEquasilinear}, and discarding all the high-order 
quantities, we obtain the linearized HME as
\begin{equation}
    \bD(\bLambda_0\bw_0)\bLambda_0\pd{\bar{\bw}}{\bar{t}}\frac{\sqrt{\theta_0}}{L} +
    \sum_{d=1}^D\bM(\bLambda_0\bw_0)\bD(\bLambda_0\bw_0)\bLambda_0\pd{\bar{\bw}}{\bar{x}_d}\frac{1}{L}
    =\bQ(\bLambda_0\bw_0)\bLambda_0\bar{\bw},
\end{equation}
where $\bS(\bLambda_0\bw_0)=0$ is applied and $\bQ=\pd{\bS}{\bw}$. 
Let $\bLambda_1=\sum_{|\alpha|\leq
M}\rho_0\theta_0^{|\alpha|/2}\bE_{\alpha,\alpha}$, then some
simplifications yield
\begin{equation}
    \bar{\bD}\pd{\bar{\bw}}{\bar{t}}+
    \sum_{d=1}^D\bar{\bM}_d\bar{\bD}\pd{\bar{\bw}}{\bar{x}_d}=
    \bar{\bQ}\bar{\bw},
\end{equation}
where 
\begin{equation}\label{eq:barDMQ}
    \begin{aligned}
        \bar{\bD}&=\bLambda_1^{-1}\bD(\bLambda_0\bw_0)\bLambda_0
        =\bI-\frac{1}{2}\sum_{d=1}^D\bE_{2e_d,0},\\
        \bar{\bM}_d&=\frac{1}{\sqrt{\theta_0}}\bLambda_1^{-1}\bM_d(\bLambda_0\bw_0)\bLambda_1
        =\sum_{|\alpha|\leq M}\left(
        \bE_{\alpha,\alpha-e_d}+(1-\delta_{|\alpha|,M})(\alpha_d+1)\bE_{\alpha,\alpha+e_d}
        \right),\\
        \bar{\bQ}&=\frac{L}{\sqrt{\theta_0}}\bLambda_1^{-1}\bQ(\bLambda_0\bw_0)\bLambda_0
        = \frac{L}{\sqrt{\theta_0}}\bLambda_1^{-1}\bQ(\bLambda_0\bw_0)\bLambda_1,
    \end{aligned}
\end{equation}
where \eqref{eq:conservedquantities} is used in the last equation.
The equation \eqref{bmS} indicates
$\bar{\bD}^{-1}\bar{\bQ}=\bar{\bQ}$,
so we have
\begin{equation} \label{linearHME}
    \pd{\bar{\bw}}{\bar{t}}+\sum_{d=1}^D\bar{\bA}_d\pd{\bar{\bw}}{\bar{x}_d}
    =\bar{\bQ}\bar{\bw}, \quad\text{ with }
    \bar{\bA}_d=\bar{\bD}^{-1}\bar{\bM}_d\bar{\bD}.
\end{equation}

To investigate the linear stability of the HME \eqref{HMEmatrix} is to
study the stability of the linearized HME \eqref{linearHME}. We first
directly propose two lemmas on the properties of the linearized HME
\eqref{linearHME} and leave the proof to the following part of this
section.
\begin{lemma}\label{Qsymmetric}
    There exists a constant invertible matrix $\bT\in\bbR^{N\times
    N}$ subject to $\bT^{-1}\bar{\bM}_d\bT$, $d=1,\cdots,D$ is symmetric,
    and $\bT^{-1}\bar{\bQ}\bT$ is symmetric negative semi-definite.
\end{lemma}
\begin{lemma}\label{lem:DQD}
    Matrices $\bar{\bD}$ and $\bar{\bQ}$ satisfy
    \begin{equation}
        \bar{\bD}^{-1}\bar{\bQ}\bar{\bD}=\bar{\bQ},
    \end{equation}
    for all the four collision models, including BGK model, Shakhov
    model, ES-BGK model and binary collision model.
\end{lemma}

With the lemmas above, our main result of this section is the
following theorem.
\begin{theorem} \label{thm:linearstability}
    The HME \eqref{HMEmatrix} is linearly stable both in space and in
    time at the local Maxwellian, i.e. the linearized HME
    \eqref{linearHME} is stable both in space and in time.
\end{theorem}
\begin{proof}
  We first prove the linear stability in time. Let $\bT$ be the
  constant invertible matrix $\bT$ in Lemma \ref{Qsymmetric}, then
  $\displaystyle\sum_{d=1}^Dk_d\bT^{-1}\bar{\bM}_d\bT$ is symmetric,
  and $\bT^{-1}\bar{\bQ}\bT$ is symmetric negative semi-definite. Due
  to Lemma \ref{eigenvaluenonnegative}, each eigenvalue of the matrix
  $\displaystyle\sum_{d=1}^Dk_d\bT^{-1}\bar{\bM}_d\bT
  -\imag\bT^{-1}\bar{\bQ}\bT$
  has a non-negative imaginary part, and thus each eigenvalue of
  \[
        \sum_{d=1}^Dk_d\bar{\bA}_d-\imag\bar{\bQ}= 
        \left( \bT^{-1}\bar{\bD} \right)^{-1}
        \left(\sum_{d=1}^Dk_d\bT^{-1}\bar{\bM}_d\bT -\imag\bT^{-1}\bar{\bQ}\bT\right) 
        \left( \bT^{-1}\bar{\bD} \right) 
  \]
  has a non-negative imaginary part, i.e. $\Omega_i\geq0$. Here Lemma
  \ref{lem:DQD} is used.

  Analogously, the linear stability in space can be proved directly with
  Lemma \ref{eigenvaluek}, Lemma \ref{Qsymmetric} and Lemma
  \ref{lem:DQD}.
\end{proof}

To finish the proof of theorem \ref{thm:linearstability}, we need to
check the validity of Lemma \ref{Qsymmetric} and Lemma \ref{lem:DQD}.
Below we construct a constant invertible $\bT$ subject to
$\bT^{-1}\bar{\bM}_d\bT$, $d=1,\cdots,D$ is symmetric at first, and then
we prove that $\bT^{-1}\bar{\bQ}\bT$ is symmetric negative
semi-definite for all the four collisions and Lemma \ref{lem:DQD}.

It is easy to see that the construction of the matrix $\bT$ is not
unique. Actually, if the matrix $\bT$ satisfies the constraints in
Lemma \ref{Qsymmetric}, then for any orthogonal matrix $\bT_1$,
$\bT_1\bT$ also satisfies the constraints in Lemma \ref{Qsymmetric}.
Here, we provide a direct construction. Precisely, if we define 
\begin{equation}
    \bT=\sum_{|\alpha|\leq
    M}\frac{1}{\sqrt{\alpha!}}\bE_{\alpha,\alpha},
\end{equation}
then 
\[
    \bT^{-1}\bar{\bM}_d\bT=\sum_{|\alpha|\leq M}
    \left( \sqrt{\alpha_d}\bE_{\alpha,\alpha-e_d}
    +(1-\delta_{|\alpha|,M})\sqrt{\alpha_d+1}\bE_{\alpha,\alpha+e_d}\right),
    \quad d=1,\cdots,D
\]
is symmetric.

In Grad's expansion \eqref{fexpand}, the basis function
$\Halphau(\bxi)$ is orthogonal but not normalized. The construction of
$\bT$ here is equivalent to a normalization of the basis functions.

Lemma \ref{lem:DQD} can be directly proved if Lemma \ref{Qsymmetric}
is valid.
\begin{proof}[Proof of Lemma \ref{lem:DQD}]
  Similarly as the derivative of \eqref{bmS}, it is easy to check
  $\bar{\bD}^T\bar{\bQ}=\bar{\bQ}$.  Let
  $\bK=\frac{1}{2}\sum_{d=1}^D\bE_{2e_d,0}$, then $\bar{\bD}=\bI-\bK$
  and $\bar{\bD}^{-1}=\bI+\bK$, and thus $\bK^T\bar{\bQ}=0$.

  It is easy again to check $\bT\bK=\frac{1}{\sqrt{2}}\bK$ and
  $\bK^T\bT^{-1}=\sqrt{2}\bK^T$, 
  thus $\bK^T\bT^{-1}\bar{\bQ}\bT=0$.  If
  Lemma \ref{Qsymmetric} is valid, then $\bT^{-1}\bar{\bQ}\bT$ is
  symmetric, and thus
  $0=\bT^{-1}\bar{\bQ}\bT\bK
  =\frac{1}{\sqrt{2}}\bT^{-1}\bar{\bQ}\bK$. Since $\bT^{-1}$
  is invertible, $\bar{\bQ}\bK=0$, which indicates
  $\bar{\bQ}\bar{\bD}^{-1}=\bar{\bQ}$. This completes the proof.
\end{proof}

Now let us prove Lemma \ref{Qsymmetric}. This requires us to verify
that $\bT^{-1}\bar{\bQ}\bT$ is symmetric negative semi-definite. Due
to the definition of $\bar{\bQ}$ \eqref{eq:barDMQ}, we need only to
show that
\begin{equation}\label{eq:Qconstrain}
    \frac{L}{\sqrt{\theta_0}}\bT^{-1}\bLambda_1^{-1}\bQ(\bLambda_0\bw_0)\bLambda_1\bT
    \text{ is symmetric negative semi-definite}.
\end{equation}
We check \eqref{eq:Qconstrain} case by case for the four
collision models we are considering:
\begin{itemize}
\item BGK model: Direct calculation of \eqref{eq:def_Salpha} yields
  $S_{\alpha}^{BGK}=H(|\alpha|-2)f_{\alpha}$, thus
  \[
    \bQ^{BGK}(\bLambda_0\bw_0)=-\frac{1}{\tau}\left(\bI-\sum_{|\alpha|\leq1}\bE_{\alpha,\alpha}
    -\frac{1}{D}\sum_{i,j=1}^D\bE_{2e_i,2e_j}\right).
  \]
  It is then easy to check \eqref{eq:Qconstrain} is valid for BGK model.
\item Shakhov model: Direct calculation of \eqref{eq:def_Salpha}
  yields
  \[
    S_{\alpha}=\left\{ \begin{array}{ll}
        0,  &   |\alpha|\leq1,\\
        \dfrac{1-\Pr}{(D+2)\tau}q_i-\dfrac{f_{\alpha}}{\tau}, &
        \alpha=e_i+2e_k,~i,k=1,\cdots,D,\\
        -\dfrac{f_{\alpha}}{\tau},   &   \text{otherwise},
    \end{array} \right.
  \]
  thus 
  \[
    \begin{aligned}
    \bQ^{Shakhov}(\bLambda_0\bw_0)&=
    -\frac{1}{\tau}\left(\bI-\sum_{|\alpha|\leq1}\bE_{\alpha,\alpha}
     -\frac{1}{D}\sum_{i,j=1}^D\bE_{2e_i,2e_j}\right.\\
    &\qquad\qquad\qquad\qquad\quad~~\left.-\frac{1-\Pr}{D+2}\sum_{i,j,k=1}^D(1+2\delta_{ij})\bE_{e_i+2e_k,e_i+2e_j}\right).
    \end{aligned}
    \]
    It is easy again to check \eqref{eq:Qconstrain} is valid for
    Shakhov model.
\item ES-BGK model: Let 
\[
    G_{\alpha}=\left\{ \begin{array}{ll}
        \rho,   &   \alpha=0,\\
        0,  &   |\alpha| \text{ is odd},\\
        \frac{1-1/\Pr}{\alpha_i\rho}\sum_{d=1}^D\sigma_{id}G_{\alpha-e_i-e_d},
        &   |\alpha|\geq2, ~i=1,\cdots,D \text{ and } \alpha_i>0,
    \end{array} \right.
\]
then $S_{\alpha}^{ES-BGK}=\frac{\Pr}{\tau}(G_{\alpha}-f_{\alpha})$.
Direct calculation yields
\begin{equation}
    \bQ^{ES-BGK}(\bLambda_0\bw_0)=-\frac{\Pr}{\tau}\left(\bI-\sum_{|\alpha|\leq2}\bE_{\alpha,\alpha}
    \right)
    -\frac{1}{\tau}\left( \sum_{d=1}^D\bE_{2e_d,2e_d}
    -\frac{1}{D}\sum_{i,j=1}^D\bE_{2e_i,2e_j} \right).
\end{equation}
One then may directly show \eqref{eq:Qconstrain} is valid for ES-BGK model.
\item Binary collision model: It it clear that the symmetry of the
  matrix $\bT^{-1}\bar{\bQ}\bT$ is equivalent to
\begin{equation}\label{eq:Qalphabeta}
    \frac{\alpha!}{\theta_0^{|\alpha|}}Q_{\alpha,\beta}(\bLambda_0\bw_0)=
    \frac{\beta!}{\theta_0^{|\beta|}}Q_{\beta,\alpha}(\bLambda_0\bw_0),
    \quad |\alpha|,|\beta|\leq M,
\end{equation}
where 
\begin{equation}\label{bQab}
    Q_{\alpha,\beta}=\pd{S_{\alpha}}{w_{\beta}},
\end{equation}
and $S_{\alpha}$ is defined in \eqref{eq:def_Salpha}.  Noticing that
at the local Maxwellian
\[
\int_{\bbR^D} \bm Q(f_{eq},f_{eq}) \Healphau(\bxi)
\dd\bxi = 0,
\]
we have
\[
    Q_{\alpha,\beta}(\bLambda_0\bw_0)= \frac{\theta^{|\alpha|}}{\alpha!} 
    \pd{(\frac{\alpha!}{\theta^{|\alpha|}}
    S_{\alpha})}{w_{\beta}}\Big{|}_{\bLambda_0\bw_0}.
\]
Let 
\[
    \bar{S}_{\alpha}=\frac{\alpha!}{\theta^{|\alpha|}}S_{\alpha},
\]
then considering \eqref{eq:def_Salpha}, we have
\[  
    \bar{S}_{\alpha}=\int_{\bbR^D}\int_{\bbR^D}\int_{S_+^{D-1}}
    \Healphau(\bxi)(f'f'_1-ff_1)B(|\bxi-\bxi_1|,\sigma)
    \dd\bn\dd\bxi_1\dd\bxi.
\]
We denote the notations
$\displaystyle \int_{\bbR^D}\int_{\bbR^D}\int_{S_{+}^{D-1}}$,
$B(|\bxi-\bxi_1|,\sigma)$ and $\dd\bn\dd\bxi_1\dd\bxi$ in the last
equation by $\displaystyle\int$, $B$ and $\dd\bm{\tau}$, respectively,
hereafter for convenience. Let $\bV\in\bbR^{N+D+1}$, and
$v_{\alpha}=f_{\alpha}$, $|\alpha|\leq M$, and $v_{N+d}=u_d$, and
$v_{N+D+1}=\theta$, then $\bV$ contains all the variables in $\bw$,
together with velocity and temperature. And thus
\[
    \bQ=\pd{\bS}{\bw}=\pd{\bS}{\bV}\pd{\bV}{\bw},
    \quad \bQ(\bLambda_0\bw_0)=\frac{\theta^{|\alpha|}}{\alpha!}\pd{\bar{\bS}}{\bV}
    \pd{\bV}{\bw}\Big{|}_{\bLambda_0\bw_0},
\]
where $\bar{\bS}=(\bar{S}_{\alpha})$.

Since $\pd{f}{s}|_{\bLambda_0\bw_0}=\pd{f_{eq}}{s}$, $s \in \{u_1,
  ... u_d, \theta \}$, 
and $f'f_1'-ff_1\mid_{\bLambda_0\bw_0}=0$ hold, we have
\[
    \pd{f'f_1'-ff_1}{s}\Big{|}_{\bLambda_0\bw_0}=
    \pd{(f'f_1'-ff_1)|_{\bLambda_0\bw_0}}{s}=0,\quad
	\quad s\in \{u_1,\ldots,u_D, \theta \}.
\]
Hence, $\bar{S}_{\alpha}$  only depends on $f_{\beta}$, $|\beta|\leq M$ and
does NOT depend on $\bu$ and $\theta$. Direct calculations yield
\[
    \begin{aligned}
        \pd{\bar{S}_{\alpha}}{w_{\beta}}\Big{|}_{\bLambda_0\bw_0}&=
        \int \Healphau(\bxi)
        \left[
            \Hbetau(\bxi')\sum_{|\gamma|\leq M}f_{\gamma}\Hgammau(\bxi_1')
            +\Hbetau(\bxi_1')\sum_{|\gamma|\leq M}f_{\gamma}\Hgammau(\bxi')
            \right.\\
        &\qquad-\left.
            \Hbetau(\bxi)\sum_{|\gamma|\leq M}f_{\gamma}\Hgammau(\bxi_1)
            -\Hbetau(\bxi_1)\sum_{|\gamma|\leq M}f_{\gamma}\Hgammau(\bxi)
            \right]B\dd \bm{\tau}\Big{|}_{\bLambda_0\bw_0}\\
        &= \int \Healphau(\bxi)
        \rho\omegau(\bxi)\omegau(\bxi_1)
        L(\beta) B\dd\bm{\tau}\Big{|}_{\bLambda_0\bw_0}\\
        &= - \frac{1}{4}\int \rho\omegau(\bxi)\omegau(\bxi_1)
        L(\alpha)L(\beta) B\dd\bm{\tau}\Big{|}_{\bLambda_0\bw_0},\\
    \end{aligned}
\]
where
$L(\alpha)= \Healphau(\bxi') + \Healphau(\bxi_1') - \Healphau(\bxi) -
\Healphau(\bxi_1)$.
Here the third equality is due to the symmetry of $\bxi$, $\bxi'$ and
$\bxi_1$, $\bxi_1'$, and the fact that the collision kernel $B$
preserves its formation once exchanging the variables
$(\bxi,\bxi_1)\leftrightarrow(\bxi',\bxi_1')$ and
$(\bxi,\bxi')\leftrightarrow(\bxi_1,\bxi_1')$ (see
\cite{Struchtrup}). Obviously, we have
\begin{equation}
    \pd{\bar{S}_{\alpha}}{w_{\beta}}\Big{|}_{\bLambda_0\bw_0}
    =\pd{\bar{S}_{\beta}}{w_{\alpha}}\Big{|}_{\bLambda_0\bw_0},
\end{equation}
which indicates $\bT^{-1}\bar{\bQ}\bT$ is symmetric.
Since $\rho\omegau(\bxi)\omegau(\bxi_1)B>0$ holds,
the matrix $\bT^{-1}\bar{\bQ}\bT$ is symmetric and negative
semi-definite.
\end{itemize}

This proved Lemma \ref{Qsymmetric}, so did Theorem
\ref{thm:linearstability}.

\section{Yong's First Stability Condition}\label{sec:Yong}

Now we examine Yong's first stability condition
\cite{yong1999singular} for HME \eqref{HMEmatrix}. The equation
\eqref{eq:equilibriummanifold} indicates that the equilibrium
manifold, denoted by $\mE$ hereafter, for HME is the local
equilibrium, which is denoted by $\bw_{eq}$ in this section. Since the
momentum is conserved, flow velocity does not change the collision
term. Due to the Galilean transformation invariance of the model, the
variation in the flow velocity is only a translation of the
system. Hence, the value of the flow velocity $\bu$ does not matter
in our discussion in this section, thus we let $\bu=0$ without loss of
generality. Each state in $\mE$ can be uniquely determined by the
density $\rho$ and the temperature $\theta$, so if we let
$\bLambda_0\bw_0=\bw_{eq}$, then all the results in Section
\ref{sec:stability} are still valid. In the following, let us directly
verify Yong's first stability condition for HME:
\begin{itemize}
\item Condition 1: 
Let
\begin{equation}
    \hat{\bP}=\bI+\sum_{i=2}^D\bE_{2e_1,2e_i},
\end{equation}
then the conservation law \eqref{eq:conservedquantities} indicates
that the first $D+2$ rows of $\hat{\bP}\bQ(\bw_{eq})$ are zeros, and
the equation \eqref{eq:conservedquantities2} indicates the other
rows are full row rank. Hence, there exists an invertible
$(N-D-2)\times(N-D-2)$ matrix $\hat{\bQ}(\bw_{eq})$ such that
\[
    \hat{\bP}\bQ(\bw_{eq}) =
    \left(\begin{array}{cc}
        0 & 0\\
        0 & \hat{\bQ}(\bw_{eq})
    \end{array}\right) \hat{\bP}.
\]

\item Condition 2:
Since $\bM_d$, $d=1,\cdots,D$ only depends on $\bw_{eq}$, we have 
\begin{equation}
    \bM_d(\bw)=\bM_d(\bw_{eq})
    =\sqrt{\theta}\bLambda_1\bar{\bM}_d\bLambda_1^{-1}
    =\sqrt{\theta}\bLambda_1\bT(\bT^{-1}\bar{\bM}_d\bT)(\bLambda_1\bT)^{-1},
    \quad d=1,\cdots,D.
\end{equation}
Let 
\begin{equation}\label{eq:YongbA0}
    \bA_0(\bw)=( (\bLambda_1\bT)^{-1}\bD(\bw))^T(
    (\bLambda_1\bT)^{-1}\bD(\bw)),
\end{equation}
then 
\[
    \bA_0\bA_d=\sqrt{\theta}
    ( (\bLambda_1\bT)^{-1}\bD)^T
    (\bT^{-1}\bar{\bM}_d\bT)(\bLambda_1\bT)^{-1}\bD
\]
is symmetric, thus \eqref{eq:Yong2} holds.

\item Condition 3:
The definition of $\bar{\bD}$ and the definition of $\bar{\bQ}$
\eqref{eq:barDMQ} indicate that
\[
    \bD(\bw_{eq})=\bLambda_1\bar{\bD}\bLambda_0^{-1},
    \qquad
    \bQ(\bw_{eq})=
    \frac{\sqrt{\theta}}{L}\bLambda_1\bar{\bQ}\bLambda_1^{-1}
    = \frac{\sqrt{\theta}}{L}\bLambda_1\bT\left( \bT^{-1}\bar{\bQ}\bT
    \right)(\bLambda_1\bT)^{-1}.
\]
Direct calculation yields
\[
    \begin{aligned}
        \bD(\bw_{eq})\bQ(\bw_{eq})
        &= \frac{\sqrt{\theta}}{L}
        \bLambda_1\bar{\bD}\bLambda_0^{-1}\bLambda_1\bar{\bQ}\bLambda_1^{-1}\\
        &=\frac{\sqrt{\theta}}{L} \bLambda_1\bar{\bD}\bar{\bQ}\bLambda_1^{-1}
        =\frac{\sqrt{\theta}}{L}
        \bLambda_1\bar{\bQ}\bLambda_1^{-1}=\bQ(\bw_{eq}),\\
    \end{aligned}
\]
where the first equality is obtained by
$\bLambda_0^{-1}\bLambda_1\bar{\bQ}=\bar{\bQ}$, and the relation
$\bar{\bD}\bar{\bQ} = \bar{\bQ}$, derived in the proof of Lemma
\ref{lem:DQD}, is used in the second equality. Analogously, we
have
\[
    \begin{aligned}
        \bQ(\bw_{eq})\bD(\bw_{eq})&=
        \frac{\sqrt{\theta}}{L}
        \bLambda_1\bar{\bQ}\bLambda_1^{-1}\bLambda_1\bar{\bD}\bLambda_0^{-1}\\
        &=
        \frac{\sqrt{\theta}}{L} \bLambda_1\bar{\bQ}\bLambda_0^{-1}
        =\frac{\sqrt{\theta}}{L}
        \bLambda_1\bar{\bQ}\bLambda_1^{-1}=\bQ(\bw_{eq}),\\
    \end{aligned}
\]
due to Lemma \ref{lem:DQD} and
$\bar{\bQ}\bLambda_0^{-1}=\bar{\bQ}\bLambda_1^{-1}$.  Thus, we have
\[
    \begin{aligned}
        \bA_0(\bw_{eq})\bQ(\bw_{eq})&=
        ((\bLambda_1\bT)^{-1}\bD(\bw_{eq}))^T(
        (\bLambda_1\bT)^{-1}\bD(\bw_{eq}))\bQ(\bw_{eq})\\
        &= ((\bLambda_1\bT)^{-1}\bD(\bw_{eq}))^T(
        (\bLambda_1\bT)^{-1})\bQ(\bw_{eq})\\
        &=\frac{\sqrt{\theta}}{L}\bD^T(\bw_{eq})(\bLambda_1\bT)^{-T}(\bT^{-1}\bar{\bQ}\bT)
        (\bLambda_1\bT)^{-1}\\
        &=\frac{\sqrt{\theta}}{L}\left((\bLambda_1\bT)^{-T}(\bT^{-1}\bar{\bQ}\bT)
        (\bLambda_1\bT)^{-1}\bD(\bw_{eq})\right)^T\\
        &=\left((\bLambda_1\bT)^{-T}(\bLambda_1\bT)^{-1}\bQ(\bw_{eq})\bD(\bw_{eq})\right)^T\\
        &=\left((\bLambda_1\bT)^{-T}(\bLambda_1\bT)^{-1}\bQ(\bw_{eq})\right)^T\\
        &=(\bLambda_1\bT)^{-T}(\bT^{-1}\bar{\bQ}\bT)(\bLambda_1\bT)^{-1}.
    \end{aligned}
\]
where $-T$ stands for transposition of inverse. It is clear that this
is a symmetric matrix. Since $\bT^{-1}\bar{\bQ}\bT$ is symmetric
negative semi-definite, there exists an invertible matrix $\bP_1$
subject to
\[
    \bT^{-1}\bar{\bQ}\bT=-\bP_1^T
    \left(\begin{array}{cc}
        0 & 0\\
        0 & \bI_{N-D-2}
    \end{array}\right)\bP_1.
\]
Therefore, there exists an invertible matrix $\bP$ subject to
both \eqref{eq:Yong1} and \eqref{eq:Yong3}.
\end{itemize}

This gives us the following theorem to end this section:
\begin{theorem}\label{thm:Yong}
    HME satisfies Yong's first stability condition.
\end{theorem}


\section{Stability Analysis of OHME}
\label{sec:grad13}
In Grad's moment method, there are two groups of moment systems. One is
choosing the basis function as
\[
    \left\{\Halphau(\bxi):|\alpha|\leq M\right\},
\]
which gives us the reduced models with $20$, $35$, $56$, $84$,
$\cdots$ moments for $D=3$.  Grad's 20 moment system is the most popular
one of them. HME are globally hyperbolic regularized version of this
group of Grad's moment system. The other one is choosing the basis
function as
\[
    \left\{\Halphau(\bxi):|\alpha|\leq M-1\right\}\bigcup 
    \left\{\sum_{d=1}^D\mathcal{H}_{\alpha+2e_d}^{[\bu,\theta]}(\bxi):|\alpha|=M-2\right\},
\]
which gives us moment system with $13$, $26$, $45$, $71$, $\cdots$
moments for $D=3$. In this group, Grad's 13 moment system is
definitely the most famous one. Following \cite{Torrilhon2015}, we
called this set of moment system as ordered Grad's moment system.

As the most important Grad's moment system, Grad's 13 moment equations
\cite{Grad} draw a lot of authors' attention in the past six decades.
Due to the lack of hyperbolicity, a globally hyperbolic
regularization, similarly as that for $M$-order Grad's moment system,
is required. In \cite{framework}, the authors extended the globally
hyperbolic regularization in \cite{Fan,Fan_new} into a framework to
derive moment equations from kinetic equations. By applying the
framework on Grad's 13 moment system, the authors proposed a globally
hyperbolic 13 moment equations (HME13). In \cite{Fan2015}, the authors
applied the globally hyperbolic regularization on ordered Grad's
moment system to obtain the Ordered Hyperbolic Moment Equations(OHME),
and pointed out that $M$-th order OHME can be derived from $M$-th
order HME.

Denote $N_O$ by the number of equations of $M$-th order OHME, and let
\[
    \bPb=\sum_{|\alpha|\leq M-1}\hat{\bE}_{\alpha,\alpha}
    +\sum_{d=1}^D\sum_{|\alpha|=M-2}\hat{\bE}_{\alpha+2e_1,\alpha+2e_d},
\]
where $\hat{\bE}_{\alpha,\beta}\in\bbR^{N_O\times N}$ is the matrix with
all its entries to be $0$, except for the only entry with row index
$\alpha$ and column index $\beta$ to be $1$.
We define the diagonal matrix $\bT_O\in\bbR^{N_O\times N_O}$ as
\[
    \bT_O=\sum_{|\alpha|\leq M-1}\frac{1}{\sqrt{\alpha!}}\bE^O_{\alpha,\alpha}
    +\sum_{|\alpha|=M-2}\frac{1}{\sqrt{\sum_{d=1}^D(\alpha+2e_d)!}}\bE^O_{\alpha+2e_1,\alpha+2e_1},
\]
where $\bE^O_{\alpha,\alpha}\in\bbR^{N_O\times N_O}$ has the
same definition as $\bE_{\alpha,\alpha}$. Let 
\begin{equation}\label{eq:def_bPp}
    \bPp=\bT_O^2\bPb(\bT^2)^{-1},
\end{equation}
then OHME can be written as \cite{Fan2015}
\begin{equation}\label{eq:OHMEoriginal}
    \begin{aligned}
        \bPp\bD(\bPb^T\bPp\bw)\bPb^T \pd{\bPp\bw}{t}
        &+\sum_{d=1}^D
        \bPp\bM_d(\bPb^T\bPp\bw)\bPb^T\bPp\bD(\bPb^T\bPp\bw)\bPb^T\pd{\bPp\bw}{x_d}\\
        &\qquad\qquad\qquad=\bPp\bS(\bPb^T\bPp\bw).
    \end{aligned}
\end{equation}
Let 
\[
    \begin{aligned}
        \bw_O&=\bPp\bw, &
        \bD_O(\bw_O)&=\bPp\bD(\bPb^T\bw_O)\bPb^T,\\
        \bS_O&=\bPp\bS(\bPb^T\bw_O),    &
        \bM_{O,d}(\bw_O)&=\bPp\bM_d(\bPb^T\bw_O)\bPb^T,
        \quad d=1,\cdots,D,
    \end{aligned}
\]
then \eqref{eq:OHMEoriginal} can be reformulated as
\begin{equation}\label{eq:OHME}
    \bD_O(\bw_O)\pd{\bw_O}{t}
    +\sum_{d=1}^D\bM_O(\bw_O)\bD_O(\bw_O)\pd{\bw_O}{x_d}
    =\bS_O(\bw_O).
\end{equation}

We claim that for this system \eqref{eq:OHME}, it is linearly stable and
fulfils Yong's first stability condition, exactly the same as HME we
studied in the last sections.

Using the same linearization as in Section \ref{sec:stability} on OHME, 
we obtain the linearized OHME as
\begin{equation}\label{eq:LinearizedOHME}
    \bar{\bD}^O\pd{\bar{\bw}_O}{t}+\sum_{d=1}^D\bar{\bM}_O\bar{\bD}_O\pd{\bar{\bw}_O}{x_d}
    =\bar{\bQ}_O\bar{\bw}_O,
\end{equation}
where 
\[
    \begin{aligned}
        \bar\bw_O&=\bPp\bar\bw, &
        \bar\bD_O&=\bPp\bar\bD\bPb^T,\\
        \bar\bQ_O&=\bPp\bar\bQ\bPb^T,   &
        \bar\bM_{O,d}&=\bPp\bar\bM_d\bPb^T,
        \quad d=1,\cdots,D.
    \end{aligned}
\]
Noticing the discussion in Section \ref{sec:stability}, we can prove OHME
is also linearly stable both in space and in time at the local
Maxwellian, once Lemma \ref{Qsymmetric} and Lemma \ref{lem:DQD} are
valid for $\bar{\bD}_O$, $\bar{\bM}_{O,d}$ and $\bar{\bQ}_O$.

Actually, due to \eqref{eq:def_bPp}, we find that both
\[
    \bT_O^{-1}\bar{\bM}_{O,d}\bT_O=(\bT^{-1}\bPb^T\bT_O)^T(\bT^{-1}\bar{\bM}_d\bT)
    (\bT^{-1}\bPb^T\bT_O),\quad d=1,\dots,D,
\]
and 
\[
    \bT_O^{-1}\bar{\bQ}_O\bT_O=(\bT^{-1}\bPb^T\bT_O)^T(\bT^{-1}\bar{\bQ}\bT)
    (\bT^{-1}\bPb^T\bT_O)
\]
are symmetric matrices. Noticing here $\bT$ and $\bT_O$ are diagonal
matrices, we obtain that Lemma \ref{Qsymmetric} is valid for
$\bar{\bM}_{O,d}$ and $\bar{\bQ}_O$.

The equation $\bar{\bD}^{-1}\bar{\bQ}=\bar{\bQ}$ is valid, since the
collision operator has $D+2$ conserved quantities and all entries of
$\bar{\bD}-\bI$ are zeroes except for some entries with row and column
indices corresponding to these conserved quantities. Since $\bPp$ and
$\bPb$ only change entries with row and column indices corresponding to
$|\alpha|>M-1$, Lemma \ref{lem:DQD} is still valid for $\bar{\bD}_O$
and $\bar{\bQ}_O$. Furthermore, we have
\begin{equation}\label{eq:OHME_DQD}
    \bD_O(\bw_{eq}^O)\bQ_O(\bw_{eq}^O)\bD_O(\bw_{eq}^O)=\bQ_O(\bw_{eq}^O),
\end{equation}
where $\bQ_O=\pd{\bS_O}{\bw_O}$.
Hence, we have the following corollary.
\begin{corollary}\label{cor:OHMELS}
  The linearized OHME \eqref{eq:LinearizedOHME} is stable both in
  space and in time. OHME \eqref{eq:OHME} is linearly stable both in
  space and in time at the local Maxwellian.
\end{corollary}

Following Section \ref{sec:Yong}, here we verify that Yong's first
stability condition is satisfied for OHME, making use of the
connections \eqref{eq:OHME} between HME and OHME.
Precisely, we have the following theorem.
\begin{theorem}\label{thm:OHMESC}
    OHME satisfies Yong's first stability condition.
\end{theorem}
\begin{proof}
Let us verify all three equalities one by one:
\begin{itemize}
\item Condition 1:
Let $\bw_{eq}^O=\bPp\bw_{eq}$. Direct calculations yield
\[
    \bQ_O(\bw_O):=\pd{\bS_O(\bw_O)}{\bw_O}=\bPp\pd{\bS(\bPb^T\bw_O)}{\bw}
    \pd{\bPp^T\bw_O}{\bw_O}=\bPp\bQ(\bw)\bPb^T.
\]
Let 
\[
    \hat{\bP}_O=\bI+\sum_{i=2}^D\bE^O_{2e_1,2e_i},
\]
then we have $\hat{\bP}_O\bPp=\bPp\hat{\bP}$ and 
$\hat{\bP}_O\bPb=\bPb\hat{\bP}$, and thus
\[
    \hat{\bP}_O\bQ_O(\bw_{eq}^O)=\bPp\hat{\bP}\bQ(\bw_{eq}^O)\bPb^T=
    \bPp \left(\begin{array}{cc}
        0 & 0\\
        0 & \hat{\bQ}(\bw_{eq})
    \end{array}\right) \bPb^T\hat{\bP}_O
    =\left(\begin{array}{cc}
        0 & 0\\
        0 & \hat{\bQ}_O(\bw_{eq})
    \end{array}\right) \hat{\bP}_O,
\]
where $\hat{\bQ}_O\in\bbR^{(N_O-D-2)\times(N_O-D-2)}$ is an invertible
matrix.
\item Condition 2:
Let $\bLambda_1^O=\bPp\bLambda_1\bPb^T$, then one is easy to see that
\[
    \bM_{O,d}=\bPp\bM_d\bPp^T=\sqrt{\theta}\bLambda_1^O\bT_0
    \bPb^T(\bT^{-1}\bar{\bM}_d\bT)\bPb(\bLambda_1^O\bT_O)^{-1}.
\]
Let
\[
    \bA_0^O(\bw_O)=
    =( (\bLambda_1^O\bT_O)^{-1}\bD_O(\bw_O))^T(
    (\bLambda_1^O\bT_O)^{-1}\bD_O(\bw_O)),
\]
then 
\[
    \bA_0^O\bA_d^O=\sqrt{\theta}
    ( (\bLambda_1^O\bT_O)^{-1}\bD_O)^T\bPb
    (\bT^{-1}\bar{\bM}_{O,d}\bT)\bPb^T(\bLambda_1^O\bT_O)^{-1}\bD_O
\]
is symmetric, thus \eqref{eq:Yong2} holds.

\item Condition 3:
Noticing \eqref{eq:OHME_DQD}, we obtain that
\[
    \begin{aligned}
        \bA_0^O(\bw^O_{eq})\bQ_O(\bw^O_{eq})&=
        ((\bLambda_1^O\bT_O)^{-1}\bD_O(\bw^O_{eq}))^T(
        (\bLambda^O_1\bT_O)^{-1}\bD_O(\bw^O_{eq}))\bQ^O(\bw^O_{eq})\\
        &= ((\bLambda^O_1\bT_O)^{-1}\bD_O(\bw^O_{eq}))^T
        (\bLambda^O_1\bT_O)^{-1}\bQ_O(\bw^O_{eq})\\
        &=\frac{\sqrt{\theta}}{L}\bD_O^T(\bw^O_{eq})(\bLambda^O_1\bT_O)^{-T}\bPb(\bT^{-1}\bar{\bQ}\bT)\bPb^T
        (\bLambda_1^O\bT_O)^{-1}\\
        &=\frac{\sqrt{\theta}}{L}\left((\bLambda_1^O\bT_O)^{-T}\bPb(\bT^{-1}\bar{\bQ}\bT)\bPb^T
        (\bLambda_1^O\bT^O)^{-1}\bD_O(\bw^O_{eq})\right)^T\\
        &=\left((\bLambda_1^O\bT_O)^{-T}(\bLambda_1^O\bT_O)^{-1}\bQ_O(\bw^O_{eq})\bD_O(\bw^O_{eq})\right)^T\\
        &=\left((\bLambda^O_1\bT_O)^{-T}(\bLambda^O_1\bT_O)^{-1}\bQ_O(\bw^O_{eq})\right)^T\\
        &=(\bLambda_1^O\bT_O)^{-T}\bPb(\bT^{-1}\bar{\bQ}\bT)\bPb^T(\bLambda_1^O\bT_O)^{-1}
    \end{aligned}
\]
is symmetric. Analogous to that in Sec. \ref{sec:Yong}, there exists 
an invertible matrix $\bP_O$ subject to both \eqref{eq:Yong1} and
\eqref{eq:Yong3}. 
\end{itemize}
This is the end of the proof.
\end{proof}


\section{Conclusion} \label{sec:conclusion}
The linear stability at the local equilibrium of both HME and OHME has
been proved with commonly used approximate collision terms, and
particularly with Boltzmann's binary collision model. Since HME
and OHME contain almost all hyperbolic regularized Grad's moment
system, the linear stability of almost all Grad-type moment system is
clarified.

Yong's first stability condition is essential to the existence of the
solution of nonlinear first-order hyperbolic with stiff source term.
The positive results in this paper may be helpful for the future study
on the existence of the solution of HME and OHME.

The linearized equation of HME is same as that of Grad's moment
equations at the local equilibrium, so the linear stability
at the local equilibrium
can be shared with the Grad's moment equations. However,
for Grad's moment equations, 
due to the lack of the hyperbolicity, even in the neighborhood of the
local equilibrium, the linear stability can not ensure the existence
of the solution. What's more, Yong's stability condition is stronger
than linear stability, which is satisfied by HME, but not 
Grad's moment equations.   

\section*{Acknowledgements}
Y. Di was supported by the National Natural Science Foundation of
China (Grant No. 11271358). Y. W. Fan was supported in part by the
National Natural Science Foundation of China (Grant
No. 91434201). R. Li was supported in part by the National Natural
Science Foundation of China (Grant No. 91330205, 11421110001,
11421101 and 11325102).


\bibliographystyle{plain}
\bibliography{../article}

\end{document}